\documentclass[letterpaper, 10 pt, conference]{ieeeconf}
\IEEEoverridecommandlockouts                              
\usepackage{cite} 
 \usepackage[inkscapearea=page]{svg}
\usepackage{graphicx}      
\usepackage{amsmath}
\usepackage{amssymb}
\usepackage{subcaption}
\usepackage{flushend}

\usepackage{amsthm}
\usepackage{mathrsfs}
\usepackage{empheq}
\usepackage{tikz}
\usepackage[utf8]{inputenc}
\usepackage{pgfplots} 
\usepackage{pgfgantt}
\usepackage{pdflscape}
\pgfplotsset{compat=newest} 
\pgfplotsset{plot coordinates/math parser=false}
\usepackage{tikzscale}
\usepackage{standalone}
\usetikzlibrary{shapes,arrows}
\usetikzlibrary{arrows,calc,positioning,fit}

\usepackage{stfloats}

\newtheorem{theorem}{Theorem}
\newtheorem{definition}{Definition}

\newtheorem{remark}{Remark}
\newtheorem{lemma}{Lemma}
\newtheorem{corollary}{Corollary}

\newtheorem{example}{Example}

\let\Im\relax
\let\Re\relax
\DeclareMathOperator{\Im}{Im}
\DeclareMathOperator{\Re}{Re}

\DeclareMathOperator{\sgn}{sgn}

\title{\LARGE \bf On phase in scaled graphs}
\author{Sebastiaan van den Eijnden, Chao Chen, Koen Scheres, Thomas Chaffey, Alexander Lanzon
\thanks{}
\thanks{Sebastiaan van den Eijnden (s.j.a.m.v.d.eijnden@tue.nl) is with the Department
of Mechanical Engineering, Eindhoven University of Technology, 5600 MB Eindhoven, The Netherlands. Koen Scheres is with the Department of Electrical Engineering (ESAT-STADIUS), KU Leuven, Leuven 3001, Belgium. His work was supported by the European Research Council under the Advanced ERC Grant Agreement SpikyControl n. 101054323. Thomas Chaffey is with the School of Electrical and Computer Engineering, University of Sydney, NSW 2008, Australia.  Chao Chen and Alexander Lanzon are with the Department of Electrical and Electronic Engineering, School of Engineering, University of Manchester, Manchester M13 9PL, UK.}}

\begin{document}

\maketitle
\pagestyle{empty}
\thispagestyle{empty}

\begin{abstract}
    The scaled graph has been introduced recently as a nonlinear extension of the classical Nyquist plot for linear time-invariant systems. In this paper, we introduce a modified definition for the scaled graph, termed the signed scaled graph (SSG), in which the phase component is characterized by making use of the Hilbert transform. Whereas the original definition of the scaled graph uses \emph{unsigned} phase angles, the new definition has \emph{signed} phase angles which ensures the possibility to differentiate between phase-lead and phase-lag properties in a system. Making such distinction is important from both an analysis and a synthesis perspective, and helps in providing tighter stability estimates of feedback interconnections. We show how the proposed SSG leads to intuitive characterizations of \emph{positive real} and \emph{negative imaginary} nonlinear systems, and present various interconnection results. We showcase the effectiveness of our results through several motivating examples.
\end{abstract}


\section{Introduction}
Graphical analysis of feedback systems through Nyquist and Bode diagrams is fundamental in classical and modern control theory \cite{Skogestad, Astrom10}. Both diagrams aid in visualising the behaviour of single-input single-output (SISO) linear time-invariant (LTI) systems in the complex plane through the notions of \emph{gain} and \emph{phase}. These tools underlie robustness margins and loop-shaping design, which are essential concepts in control engineering practice. Notably, Nyquist and Bode diagrams are only defined for LTI systems, owing to the lack of a suitable definition of phase for nonlinear systems.  

To provide a bridge between graphical analysis of linear and nonlinear systems, the notion of \emph{scaled graph} (SG) has been introduced recently in \cite{Ryu22, Chaffey21, Chaffey22, Chaffey23, Eijnden24, Pates21} as a means to visualise the behaviour of nonlinear systems in the complex plane, analogous to the classical Nyquist diagram for LTI systems. In essence, the SG abstracts the input-output behaviour of a nonlinear system in terms of a collection of complex numbers, which carry both gain and phase information of the system's input-output trajectories. Gain is characterized in terms of the classical notion of the $\mathcal{L}_2$-gain \cite{vdSchaft17}, i.e., the ratio of output energy to input energy. Phase is characterised through the notion of the \emph{singular angle}, see, e.g., \cite{Wielandt67} and \cite{Chen21}, which can be seen as an extension of the angle between vectors in Euclidean space. The SG offers a new graphical method for verifying stability and robustness margins of nonlinear feedback systems, and entails classical input-output results such as the celebrated passivity, small-gain and conicity theorems \cite{Zames66, Desoer09}.      

When comparing the SG to the Nyquist diagram, however, there are essential differences in systems analysis and controller design. Although both the SG and the Nyquist diagram of an LTI system are conjugate symmetric about the real axis, in a Nyquist diagram a clear distinction can be made between the symmetric parts (related to positive and negative frequencies). That is, one can distinguish between parts on the Nyquist diagram that induce phase lag between $[-180, 0]$ degrees (negative imaginary part), and phase lead between $[0, +180]$ degrees (positive imaginary part). Making such a distinction in the SG notion is not possible. The reason for this difference is due to the use of the singular angle for characterizing phase in the SG. The singular angle is an \emph{unsigned} quantity, making it impossible to distinguish lead from lag in the SG. This introduces a certain conservatism when using SG for systems analysis. For example, it is impossible to distinguish LTI phase-lead filters from LTI phase-lag filters in the SG characterization (as we show in this paper), and thus the stabilizing mechanism of phase lead is placed on an equal footing with the potentially destabilizing mechanism of phase lag. Clearly, being able to differentiate between elements that generate phase lead or phase lag is crucial in controller design -- think, for instance, about proportional-integral-derivative (PID) filters where the sole purpose of the derivative action is to provide phase lead for meeting robust stability margins \cite{Astrom95}.  

Given the important and different roles of phase lead and phase lag in feedback systems, in this paper, we reconsider the current definition of the SG. In particular, we propose a different notion of phase in the SG that is based on using the Hilbert transform of a signal \cite{King09}. Use of the Hilbert transform for defining nonlinear system phase was pioneered in \cite{Chen23}. We combine sign information of the Hilbert transform with the singular angle to provide a signed phase, and, therefore, allow to distinguish phase lead from phase lag in SG analysis.

In line with the above, the main contributions of this paper are as follows. We propose a new definition of the scaled graph that is different from the existing definition, in the sense that phase lead can be distinguished from phase lag. We refer to this new definition as the \emph{signed scaled graph} (SSG). We provide a feedback interconnection result that reduces potential conservatism in the current SG analysis. Our new SSG allows for a natural definition of \emph{positive real} and \emph{negative imaginary} nonlinear systems, which are intimately related to phase properties. As a second contribution, we introduce these definitions from the SSG. While our characterization of positive real systems is consistent with the definitions found in the literature \cite{Khalil02}, our definition of negative imaginary systems is different. In particular, the classical definitions of negative imaginary systems involve time-derivatives or time-integrals of the input/output signals \cite{Ghallab18, Patra11, Lanzon23, Zhao22}, whereas our new definition exploits Hilbert transformed signals. We believe this definition to be more natural, as the Hilbert transform provides a pure $90$-degree phase shift to the Fourier spectrum of a signal, whereas differentiation/integration provide such phase shifts but also change the amplitude characteristics. We show that the interconnection result based on the SSG recovers classical passivity interconnection results, as well as an alternative to negative imaginary interconnection results \cite{Lanzon08}, the latter of which could not be derived from the original SG definition in \cite{Chaffey23}. The ideas put forward in this paper enrich the SG theory and provide next steps in graphical nonlinear system analysis and design.      

The remainder of this paper is structured as follows. In Section~\ref{sec:preliminaries} we provide preliminaries on SGs and examples that motivate the need for rethinking phase in the original SG definition. Section~\ref{sec:newSG} presents our first main contribution in the form of the signed SG definition, as well as an interconnection result for stability analysis. In Section~\ref{sec:NL} we introduce and discuss new definitions for positive real and negative imaginary systems, which forms our second contribution. Our results are supported by an example. Section~\ref{sec:conclusion} concludes the paper.


\textbf{Notation.} 
For $T\geq 0$ and signals $u, y: [0,T] \to \mathbb{R}^n$, let 
\begin{equation*}\label{eq:int}
\langle u,y\rangle_T = \int_{0}^T u(t)^\top y(t) \,dt \:\:\textup{ and } \:\: \|u\|_T = \sqrt{\langle u,u\rangle_T}.
\end{equation*} 
When $u,y:[0,\infty)\rightarrow \mathbb{R}^n$, we adopt the standard notation $$\langle u,y \rangle = \int_{0}^\infty u(t)^\top y(t) \,dt\:\: \text{ and } \:\:\|u\| = \sqrt{\langle u,u\rangle}. $$
The space of signals that are square-integrable over any finite time interval $[0, T]$, i.e., $\|u\|_T < \infty$, is denoted by $\mathcal{L}_{2e}$. We let $\mathcal{L}_2$ denote the space of square-integrable signals on the time axis $[0, \infty)$. Furthermore, for signals $u,y :(-\infty, \infty) \to \mathbb{R}^n$ we denote 
\begin{equation}
    \pmb{\langle} u,y\pmb{\rangle} = \int_{-\infty}^\infty u(t)^\top y(t)\,dt,
\end{equation}
and let the space of corresponding square-integrable signals be denoted by $\mathcal{L}_{2}(-\infty, \infty)$.


We denote the distance between two sets $A, B \subset \mathbb{C}$ by $\textup{dist}(A,B) = \inf_{a \in A, b \in B}|a-b|$. The real and imaginary parts of a complex number $z \in \mathbb{C}$ are denoted by $\Re\left\{z\right\}$ and $\Im\left\{z\right\}$, respectively. For a nonzero complex number $z \in \mathbb{C}$ in polar form $|z|e^{j \arg (z)}$, the magnitude is denoted by $|z|$ and the angle is denoted by $\arg(z) \in [-\pi,\pi)$. If $z = 0$ or $z = \infty$, then $\arg(z)$ is undefined. We denote the complex conjugate of a complex number $z = |z|e^{j \arg (z)}$ by $z^*$, i.e., $z^* = |z|e^{-j \arg (z)}$. The inverse of a complex number $z$ is obtained by inverting its magnitude and flip the sign of its phase, i.e., $z^\dagger = (1/|z|)e^{-j\arg (z)}$, such that $z z^\dagger = 1$.

\section{Preliminaries and motivation}\label{sec:preliminaries}
In this section, we provide preliminaries on SGs, and present some examples that highlight the need for rethinking the phase in the original SG definition.  

\subsection{Scaled graphs}
A system is a map $H : \mathcal{L}_2 \to \mathcal{L}_2$. Let $u\in \mathcal{L}_2$ denote the input and $y \in \mathcal{L}_2$ denote a corresponding output $y \in H(u)$.  We make the standing assumption that $H(0)=0$. Define for all $u, y\neq 0$ the \emph{gain} $\rho(u,y) \in (0, \infty)$ and the \emph{phase} $\theta(u,y) \in [0,\pi]$ by
\begin{equation}\label{eq:gain_phase}
    \rho(u,y) = \frac{\|y\|}{\|u\|} \quad \textup{and} \quad \theta(u,y)=\arccos \frac{\langle u,y\rangle}{\|u\|\|y\|},
\end{equation}
respectively, and $\rho(u,y) = \infty$ and $\theta(u,y) = 0$ when $u=0$. The scaled graph of $H$ is defined as \cite{Ryu22, Chaffey23}
\begin{equation}\label{eq:SG}
    \textup{SG}(H) = \left\{\rho(u,y)e^{\pm j\theta(u,y)} \mid u \in \mathcal{L}_2, \:y \in H(u)\right\}.
\end{equation}
Each input-output pair of $H$ results in a complex number whose magnitude and argument provide gain and phase information of $H$. The inverse of the SG in \eqref{eq:SG} is denoted by $\textup{SG}^\dagger (H)$ and is obtained by swapping the role of the input and output of $H$. Note that the phase $\theta(\cdot,\cdot)$ in \eqref{eq:gain_phase} takes values in $[0,\pi]$ and, therefore, is an unsigned quantity. By using both $+\theta(\cdot,\cdot)$ and $-\theta(\cdot,\cdot)$ in \eqref{eq:SG}, $\textup{SG}(H)$ is symmetric around the real axis and so is its inverse $\textup{SG}^\dagger(H)$.

In the special case that $H$ is linear and time-invariant with transfer function $H(s)$, and analytical inputs of the form $u(t) = e^{j\omega t}$, the gain and phase in \eqref{eq:gain_phase} are given by
\begin{equation}
    \rho(u,y) = |H(j\omega)| \quad \textup{and} \quad \theta(u,y) = |\arg (H(j\omega))|.
    \end{equation}
Thus, the gain and phase of the scaled graph in \eqref{eq:SG} are intimately related to the classical notions of the gain and phase of a transfer function. We refer the reader to \cite[Section~V]{Chaffey23} and \cite{Pates21} for further connections between the SG and the Nyquist diagram of an LTI system.  

The SG provides an elegant means for analysis of the feedback interconnection of two casual and bounded (i.e. open-loop stable) systems $H_1:\mathcal{L}_{2} \to \mathcal{L}_{2}$ and $H_2: \mathcal{L}_{2} \to \mathcal{L}_{2}$ as shown in Fig.~\ref{fig:FB}. Here, $w \in \mathcal{L}_{2}$ is an external signal, and $u_1, u_2, y_1, y_2\in \mathcal{L}_{2e}$ are internal signals. Before providing a feedback stability result based on SGs, we recall formal definitions of feedback well-posedness and stability.   

\begin{definition}[Well-posedness]
    Let $H_1:\mathcal{L}_{2} \to \mathcal{L}_{2}$ and $H_2: \mathcal{L}_{2} \to \mathcal{L}_{2}$. The feedback interconnection in Fig.~\ref{fig:FB} is said to be well-posed if the map $u_1 \mapsto w$ has a causal inverse on $\mathcal{L}_{2e}$. 
  \end{definition}

\begin{definition}[Feedback stability]
    Let $H_1:\mathcal{L}_{2} \to \mathcal{L}_{2}$ and $H_2: \mathcal{L}_{2} \to \mathcal{L}_{2}$. The feedback interconnection in Fig.~\ref{fig:FB} is said to be stable if it is well-posed, and inputs $w\in \mathcal{L}_2$ are mapped to outputs $u_1 \in \mathcal{L}_2$. We say it is finite-gain stable if it is stable and in addition there exists $\gamma >0$ such that $\|u_1\| \leq \gamma \| w\|$.
\end{definition}

The next result, adopted from \cite[Theorem 4]{Eijnden24}, provides a stability theorem for the interconnection in Fig.~\ref{fig:FB} that is based on graphical separation of $\textup{SG}(H_1)$ and $\textup{SG}(H_2)$ in $\mathbb{C}$.

\begin{figure}[htbt!]
\centering
\setlength{\unitlength}{1.2mm}
\begin{picture}(50,25)
\thicklines \put(0,20){\vector(1,0){8}} \put(10,20){\circle{4}}
\put(12,20){\vector(1,0){8}} \put(20,15){\framebox(10,10){$H_1$}}
\put(30,20){\vector(1,0){16}} \put(40,20){\line(0,-1){15}}
\put(40,5){\vector(-1,0){10}}  
 \put(20,0){\framebox(10,10){$H_2$}}
\put(20,5){\line(-1,0){10}} \put(10,5){\vector(0,1){13}}
\put(13,0){\makebox(5,5){$y_2$}} \put(32,20){\makebox(5,5){$y_1$}}
\put(0,20){\makebox(5,5){$w$}}  
\put(13,20){\makebox(5,5){$u_1$}} \put(32,0){\makebox(5,5){$u_2$}}
\put(10,10){\makebox(6,10){$-$}}
\end{picture} 
  \caption{Negative feedback interconnection.}
        \label{fig:FB}
\end{figure}
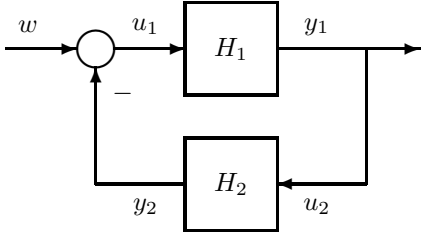    

\begin{theorem}\label{th:SG}
Consider a pair of stable systems $H_1:\mathcal{L}_2 \to \mathcal{L}_2$ and $H_2:\mathcal{L}_2 \to \mathcal{L}_2$, and suppose that the feedback interconnection in Fig.~\ref{fig:FB} between $H_1$ and $\tau H_2$ is well-posed for all $\tau \in (0,1]$. If there exists $r >0$ such that for all $\tau \in (0,1]$ 
\begin{equation}\label{eq:cond}
   \textup{dist}(\textup{SG}^\dagger(H_1), \textup{SG}(-\tau H_2))\geq  r,
\end{equation}
then the feedback interconnection is finite-gain stable. 
\end{theorem}


 The appealing aspect of Theorem~\ref{th:SG} is the fact that condition \eqref{eq:cond} can be checked graphically, and is reminiscent of a classical Nyquist stability test. Recall that due to the phase in \eqref{eq:gain_phase}, the scaled graph and its inverse are symmetric in $\mathbb{C}$ around the real axis. In the next section we will provide an example that illustrates possible conservatism in Theorem~\ref{th:SG} coming from this inherent symmetry.


\subsection{Motivating example}\label{sec:examples}

The next example highlights the implications of conjugate symmetry in scaled graphs for feedback stability analysis. 

\begin{example}[Lead and lag filters]\label{ex:0}
\normalfont Consider a first-order LTI lead filter and a first-order LTI lag filter, given by
\begin{equation}\label{eq:Hll}
    H_{\textup{lead}}(s) = \frac{s}{s+1} \quad \textup{ and } \quad H_{\textup{lag}}(s) = \frac{1}{s+1},
\end{equation}
respectively. The Nyquist diagram of both $H_{\textup{lead}}(s)$ as well as $H_{\textup{lag}}(s)$ is the circle in $\mathbb{C}$ given by $\left\{z \in \mathbb{C} \mid |z-0.5| = 0.5\right\}$. This circle precisely coincides with the scaled graphs of both $H_{\textup{lead}}(s)$ and $H_{\textup{lag}}(s)$, i.e., $\textup{SG}(H_{\textup{lead}}) = \textup{SG}(H_{\textup{lag}})$. Hence, we cannot distinguish between an LTI lead filter and an LTI lag filter based on their scaled graphs alone. 

To show the implications of $\textup{SG}(H_{\textup{lead}}) = \textup{SG}(H_{\textup{lag}})$ in feedback stability analysis, consider the LTI system
\begin{equation}\label{eq:G}
    H_1(s) = \frac{k}{(s+1)^2}
\end{equation}
with $k >0$. From the Nyquist criterion we infer that the negative feedback interconnection of $H_1$ and $H_2 = H_{\textup{lead}}$ is stable for all values $k > 0$, whereas the negative feedback interconnection of $H_1$ and $H_2 = H_{\textup{lag}}$ is stable for all $k <8$. 

The scaled graph of the system in \eqref{eq:G} is given by the region with perimeter $p(\varphi) = k\left(\cos\left(\frac{\varphi}{2}\right)\right)^2e^{-j\varphi}$, where $\varphi \in [-\pi, \pi]$ (see \cite{Chaffey22}). Since the inverse of the scaled graphs of both the lead and lag filters in \eqref{eq:Hll} are given by
\begin{equation}\label{eq:SGi}
    \textup{SG}^\dagger(H_{\textup{lead}}) =  \textup{SG}^\dagger(H_{\textup{lag}}) = \left\{z \in \mathbb{C} \mid \Re\:\{z\} = 1\right\},
\end{equation}
it follows from application of Theorem~\ref{th:SG} that the negative feedback interconnection of $H_1$ in \eqref{eq:G} with either $H_2=H_{\textup{lead}}$ or $H_2=H_{\textup{lag}}$ is stable if
$\Re\left\{p(\varphi)\right\} > -1$ for all $\varphi \in [-\pi,\pi]$, which is true for $k < 8$. Hence, when $H_2 = H_{\textup{lag}}$, Theorem~\ref{th:SG} introduces no conservatism, whereas when $H_2=H_{\textup{lead}}$, Theorem~\ref{th:SG} is conservative. This conservatism comes from the fact that in \eqref{eq:SGi} no distinction between lead and lag can be made, and, therefore, Theorem~\ref{th:SG} accounts for the worst-case situation, i.e., it accounts for phase lag as a destabilizing mechanism. $\hfill$\qed
\end{example}

One could argue that in the example discussed, stability could also be verified by evaluating $\textup{SG}(L)$ for the loop system $L=H_1H_2$ with respect to the point $-1+j0$ directly. For LTI systems, Theorem~\ref{th:SG} then coincides with the classical Nyquist test. However, in controller design and robustness analysis, there is good reason for considering the scaled graphs of $H_1$ and $H_2$ separately. Being able to differentiate between phase lead and phase lag is thus of importance in analysis and design. 

\section{Signed scaled graphs}\label{sec:newSG}
Motivated by the previous examples, in this section we introduce an alternative construction of the scaled graph. Key in this construction is the Hilbert transform of a signal, which we will discuss in more detail next. 

\subsection{The Hilbert transform}\label{sec:Hilbert}
The Hilbert transform $\mathcal{H}:\mathcal{L}_2(-\infty,\infty)\rightarrow\mathcal{L}_2(-\infty,\infty)$ of a signal $f\in \mathcal{L}_2(-\infty,\infty)$ is defined as the convolution of $f(t)$ with the function ${1}/{(\pi t)}$, that is, 
\begin{equation}\label{eq:H0}
   \hat{f}(t) = \mathcal{H}\{f\}(t) = \frac{1}{\pi t} *f(t) =\frac{1}{\pi} \int_{-\infty}^{\infty} \frac{f(\tau)}{t-\tau}\,d\tau. 
\end{equation}
The integral in \eqref{eq:H0}  is improper due to the pole at $\tau = t$, and is evaluated in the sense of the Cauchy principal value \cite{King09}. Moreover, $\hat{f}\in \mathcal{L}_2(-\infty, \infty)\supset \mathcal{L}_2$. 

Let the embedding operator $\mathcal{E}:\mathcal{L}_2\rightarrow\mathcal{L}_2(-\infty,\infty)$ be defined by 
\[ \mathcal{E}(u)(t) = \begin{cases}
                       u(t) & \text{when } t\geq 0, \\
                       0 & \text{otherwise}
                    \end{cases} \]
for any $u\in\mathcal{L}_2$. It is easy to see that $u\in\mathcal{L}_2$ implies $\mathcal{E}(u)\in\mathcal{L}_2(-\infty,\infty)$. Then, the Hilbert transform of the embedding of a signal $u\in\mathcal{L}_2$ is given by
\begin{equation}\label{eq:H}
   \hat{u}(t) = \mathcal{H}\{\mathcal{E}(u)\}(t) = \frac{1}{\pi t} *\mathcal{E}(u)(t) =\frac{1}{\pi} \int_{0}^{\infty} \frac{u(\tau)}{t-\tau}\,d\tau
\end{equation}
and $\hat{u}\in\mathcal{L}_2(-\infty,\infty)\supset \mathcal{L}_2$. 

The intuition behind the Hilbert transform can be developed by looking at the Fourier transform\footnote{Recall that the Fourier transform $\mathcal{F}:\mathcal{L}_2(-\infty,\infty)\rightarrow\mathcal{L}_2(j\mathbb{R})$ of a signal $f\in\mathcal{L}_2(-\infty,\infty)$ is given by $\tilde{f}(j\omega)=\mathcal{F}\{f\}(j\omega)=\int_{-\infty}^{\infty}f(t)e^{j\omega t}\,dt$ where $\tilde{f}\in\mathcal{L}_2(j\mathbb{R})$.}
$\mathcal{F}\{\cdot\}$ of the convolution kernel $1/(\pi t)$ in \eqref{eq:H}, which is given by $\mathcal{F}\{1/(\pi t)\} =\int_{-\infty}^{\infty}\frac{1}{\pi t}e^{j\omega t}\,dt = -j \sgn(\omega)$, where $\sgn(\cdot)$ is the sign function. As such, the Fourier transform of the Hilbert transformed signal $\hat{u}(t)$ is given by $\tilde{\hat{u}}(j\omega) = -j \sgn(\omega)\tilde{u}(j\omega)$ where 
$\tilde{\hat{u}}(j\omega)=\int_{-\infty}^{\infty}{\hat{u}}(t)e^{j\omega t}\,dt$ and $\tilde{u}(j\omega)=\int_{-\infty}^{\infty}\mathcal{E}(u)(t)e^{j\omega t}\,dt =\int_{0}^{\infty}u(t)e^{j\omega t}\,dt$. The Hilbert transform provides a pure $\pi/2$ phase shift to the frequency components in $u(t)$ and is often utilized in mathematics and signal processing to generate an analytic signal $u_a(t) = \mathcal{E}(u)(t)+j \hat{u}(t)$ that has no negative frequency components. 

For SISO LTI systems, the Hilbert transform is intimately related to the imaginary part of the system's frequency-response-function. Consider a stable SISO LTI system $y = Hu$ with transfer function $H(s)$. When evaluating the inner products of the output with the input and Hilbert transformed input, we find by application of Plancherel's theorem
\begin{align}\label{eq:uyl}
\begin{split}
    \langle u,y\rangle &= \int_0^{\infty} u(t) y(t)\, dt =\int_{-\infty}^{\infty} \mathcal{E}(u)(t) \mathcal{E}(y)(t)\, dt \\
    & =\frac{1}{2\pi} \int_{-\infty}^{\infty} H(j\omega)|\tilde{u}(j\omega)|^2 \,d\omega \\
    & = \frac{1}{\pi} \int_{0}^{\infty} \Re\left\{H(j\omega)\right\}|\tilde{u}(j\omega)|^2 \,d\omega,
    \end{split}
\end{align}
and 
\begin{align}\label{eq:uhyl}
\begin{split}
    -\pmb{\langle} \hat{u},\mathcal{E}(y)\pmb{\rangle} & =  - \int_{-\infty}^\infty \hat{u}(t) \mathcal{E}(y)(t) \,d t \\
    & =\frac{1}{2\pi} \int_{-\infty}^{\infty} -j\sgn(\omega)H(j\omega)|\tilde{u}(j\omega)|^2 \,d\omega \\
    & = \frac{1}{\pi} \int_{0}^{\infty} \Im\left\{H(j\omega)\right\}|\tilde{u}(j\omega)|^2 \,d\omega,
    \end{split}
\end{align} 
where we used the conjugate symmetry properties of $\sgn(\omega)$ and $H(j\omega)$, and the fact that $\mathcal{E}(y)(t)=0$ for all $t<0$. Hence, in essence, $\langle u,y\rangle$ provides information on the real part of $H(j\omega)$, whereas $\pmb{\langle} \hat{u},\mathcal{E}(y)\pmb{\rangle}$ provides information on the imaginary part. Combining the two quantities gives us full information in terms of gain and phase. Given these observations for LTI systems, it is natural to consider the Hilbert transform within the nonlinear input-output setting as well. 



\subsection{Signed scaled graphs through the Hilbert transform}
In view of the above, we consider an alternative construction of the scaled graph. Specifically, given any signal $u\in \mathcal{L}_2$,  for a system $H:\mathcal{L}_2 \to \mathcal{L}_2$, we define its gain similar to \eqref{eq:gain_phase}, that is, $\rho(u,y)=\|y\|/\|u\|$ for all $u,y \neq 0$, and $\rho(u,y)= \infty$ for $u = 0$. Different from \eqref{eq:gain_phase}, we now define the phase $ \varphi(u, y) \in [-\pi, \pi]$ of $H$ as

{\begin{equation}\label{eq:phase}
    \varphi(u,y) = \sgn(\pmb{\langle} \hat{u},\mathcal{E}(y)\pmb{\rangle}) \theta(u,y), 
\end{equation}
if $\pmb{\langle} \hat{u},\mathcal{E}(y)\pmb{\rangle}\neq 0$ and $\varphi(u,y) =\pm \theta(u,y) $ if $\pmb{\langle} \hat{u},\mathcal{E}(y)\pmb{\rangle}= 0$, where $\theta(u,y)$ is the (unsigned) phase given in \eqref{eq:gain_phase}.}


Different from the singular angle in \eqref{eq:gain_phase}, this notion of phase provides a \emph{signed} quantity. The \emph{signed scaled graph} of $H$ is then defined as 
\begin{equation}\label{eq:SGH}
    \textup{SSG}(H) := \left\{\rho(u,y)e^{j\varphi(u,y)} \mid u \in \mathcal{L}_2, \:y \in H(u)\right\},
\end{equation}
and its inverse $\textup{SSG}^\dagger(H)$ is obtained by swapping the role of the input and output. As such, each $z \in \textup{SSG}^\dagger(H)$ can be written as $$z(y,u) = \frac{1}{\rho(u,y)}e^{-j\varphi(u,y)},$$ which follows from the property that $\pmb{\langle} \hat{y},\mathcal{E}(u)\pmb{\rangle} = -\pmb{\langle} \hat{u},\mathcal{E}(y)\pmb{\rangle}$ \cite{Chen23}, and, therefore, by swapping the role of the input and output in $H$, the phase flips sign in line with our intuition. We furthermore denote the set $\textup{SSG}^*(H)$ as the conjugate of \eqref{eq:SGH}, i.e.,  $\textup{SSG}^*(H)=\{z^* \mid z\in \textup{SSG}(H)\}$.


Note that \eqref{eq:SGH} and its inverse are no longer symmetric around the real axis, although we could simply make them symmetric by including the complex conjugates as well. However, because \eqref{eq:phase} is a signed quantity, all relevant information is captured in \eqref{eq:SGH}. This construction is akin to the Nyquist diagram of an LTI system, where all relevant information is contained in the part for positive frequencies. It is clear that the signed scaled graph in \eqref{eq:SGH} is closely related to the original definition in \eqref{eq:SG} in the following way.

\begin{lemma}\label{prop:0}
For $H:\mathcal{L}_2 \to \mathcal{L}_2$, it holds that $\textup{SSG}(H) \subset \textup{SG}(H)$~and~$\textup{SSG}(H)\cup \textup{SSG}^*(H) = \textup{SG}(H)$. 
\end{lemma}

The proof of this lemma is straightforward. 

\begin{remark}\label{rem:1}
    To be consistent with the original definition of the scaled graph, we take both the positive and negative singular angle as the phase in case $\pmb{\langle} \hat{u},\mathcal{E}(y)\pmb{\rangle}$ is zero, i.e., $\varphi(u,y) = \pm \theta(u,y)$. This leads to excessive points in the SSG in this case, as we cannot recover more information on the phase than is provided by the singular angle. 
\end{remark}

\subsection{An interconnection result}
In this subsection we present a stability result for the feedback interconnection in Fig.~\ref{fig:FB} based on the signed scaled graph in \eqref{eq:SGH}. Before providing the result, we define the set
\begin{equation}\label{eq:MW}
    \begin{split}
       \mathcal{W}=\{&w \in \mathcal{L}_{2} \mid w = u_1+\tau y_2, \textup{ and } \exists\;\tau \in (0,1] \text{ s.t. } \\
      & \langle u_1,u_2\rangle  +\langle u_2,\tau y_2\rangle = 0, \|u_1\| - \|\tau y_2\|=0,\\
      &\pmb{\langle} \hat{u}_1, \mathcal{E}(u_2)\pmb{\rangle} \pmb{\langle} \hat{u}_2, \mathcal{E}(\tau y_2) \pmb{\rangle} < 0\}.
          \end{split}
    \end{equation}
For those feedback systems satisfying Theorem~\ref{th:SG}, this set is automatically empty. Additionally, $\mathcal{W}$ is empty for the negative feedback interconnection of two (strictly) passive systems \cite{Khalil02}, and the negative feedback interconnection of two LTI negative imaginary systems \cite{Lanzon08}.




\begin{theorem}\label{th:SGH}
Consider a pair of finite-gain stable systems $H_1:\mathcal{L}_2 \to \mathcal{L}_2$ and $H_2:\mathcal{L}_2 \to \mathcal{L}_2$ and suppose that the feedback interconnection in Fig.~\ref{fig:FB} between $H_1$ and $\tau H_2$ is well-posed for $\tau \in (0,1]$. If there exists $r >0$ such that for all $\tau \in (0,1]$ 
\begin{equation}\label{eq:cond2}
   \textup{dist}( \textup{SSG}^\dagger(H_1), \textup{SSG}(-\tau H_2))\geq  r,
\end{equation}
then the feedback interconnection is finite-gain stable for all inputs $w \in \mathcal{L}_2 \setminus \mathcal{W}$.
\end{theorem}

The proof of Theorem~\ref{th:SGH} can be found in Appendix~\ref{app:pf2} and is inspired by the idea of partitioning input-output trajectories in pairs leading up to small-gain, large-gain and small-phase pairs \cite{ChenBLX}. Theorem~\ref{th:SGH} still holds if  condition~\eqref{eq:cond2} is replaced by $  \textup{dist}( \textup{SSG}(H_1), \textup{SSG}^\dagger(-\tau H_2))\geq  r$.

Theorem~\ref{th:SGH} is similar to Theorem~\ref{th:SG} in the sense that both results are based on topological graph separation. Different from Theorem~\ref{th:SG}, however,  Theorem~\ref{th:SGH} exploits the signed scaled graph in \eqref{eq:SGH} which is not necessarily symmetric in $\mathbb{C}$ around the real axis. The latter may offer advantages in reducing conservatism in feedback stability analysis. 
{The restriction on the set of allowed inputs $w$ is a technical condition required for the proof of Theorem~\ref{th:SGH}, which we expect can be removed in future work.}

In the remaining part of this section we revisit Example~\ref{ex:0} from Section~\ref{sec:examples}, now in light of the signed scaled graph in \eqref{eq:SGH} and Theorem~\ref{th:SGH}.

\begin{example}[Lead and lag filters revisited]\label{ex:3rev}
    \normalfont Consider again the LTI lead and LTI lag filters given in \eqref{eq:Hll}. For all $\omega \in [0,\infty)$, the real parts of these filters satisfy $\textup{Re}\left\{H_{\textup{lead}}(j\omega)\right\} \geq 0$ and $\textup{Re}\left\{H_{\textup{lag}}(j\omega)\right\} \geq 0$, and their imaginary parts satisfy
\begin{align}
    &\textup{Im}\left\{H_{\textup{lead}}(j\omega)\right\} = {\omega}/{(\omega^2+1)} \geq 0, 
    \end{align}
    and $\textup{Im}\left\{H_{\textup{lag}}(j\omega)\right\} = -\textup{Im}\left\{H_{\textup{lead}}(j\omega)\right\} \leq 0$ for all $\omega \in [0,\infty]$. From this, we find that
    $$-\textup{SSG}^\dagger(H_{\textup{lead}}) \subset \left\{z \in \mathbb{C} \mid \textup{Re}\left\{z\right\} = -1, \textup{Im}\left\{z \right\} \geq 0\right\},$$
    and $$-\textup{SSG}^\dagger(H_{\textup{lag}})\subset \left\{z \in \mathbb{C} \mid \textup{Re}\left\{z\right\} = -1, \textup{Im}\left\{z \right\} \leq 0\right\}.$$ 
    It is important to note that from the frequency response functions of these systems we find that $\pmb{\langle} \hat{u},\mathcal{E}(y)\pmb{\rangle} = 0$ is true only for inputs\footnote{Note that technically speaking these inputs do not belong to $\mathcal{L}_2$, and can, therefore, be excluded.} with $\omega = 0$. In turn, this implies $\textup{Re}\left\{H_{\textup{lead/lag}}(j\omega)\right\} = |H_{\textup{lead/lag}}(j\omega)|$ and, therefore, $\langle u,y\rangle/\|u\|^2 = \|y\|/\|u\|$. Thus, in this case, those points in the signed scaled graph where $\pmb{\langle} \hat{u},\mathcal{E}(y)\pmb{\rangle} = 0$ collapse to points on the real axis. 
Consider also the plant $H_1$ in \eqref{eq:G} and note that $\textup{SSG}(H) \subset \textup{SG}(H)$. Similar as before, $\pmb{\langle} \hat{u},\mathcal{E}(y)\pmb{\rangle} = 0$ is true for $\omega = 0$ and excessive points collapse to points on the real axis. Since $\textup{SG}(H_1)$ is completely bounded by the Nyquist diagram of $H_1$, which belongs to the third and fourth quadrant in the complex plane, it readily follows that $\textup{SSG}(H_1)$ and $-\textup{SSG}^\dagger(\tau H_{\textup{lead}})$ do not intersect for any $\tau \in (0,1]$. By Theorem~\ref{th:SGH}, the negative feedback interconnection of $H_1$ and $H_2 = H_{\textup{lead}}$ is stable for all values of $k >0$. Moreover, similar as in the earlier example, $-\textup{SSG}^\dagger(\tau H_{\textup{lag}})$ does not intersect $\textup{SSG}(H_1)$ for any $\tau \in (0,1]$ provided $k < 8$. Hence, we recover the same result as in the Nyquist stability test. Finally, it can be easily shown via the closed-loop mapping that $\mathcal{W} \neq \mathcal{L}_2$.
$\hfill$\qed
\end{example}

As a final result in this section, we show that Theorem~\ref{th:SGH} is never more conservative than Theorem~\ref{th:SG}.

\begin{theorem}
    If condition \eqref{eq:cond} in Theorem~\ref{th:SG} is satisfied for all $\tau \in (0,1]$, then condition \eqref{eq:cond2} in Theorem~\ref{th:SGH} is satisfied for all $\tau \in (0,1]$. The converse is not true in general.
\end{theorem}

\begin{proof}
    This is a direct result of Lemma~\ref{prop:0} and the fact that $\mathcal{W}=\emptyset$ under condition \eqref{eq:cond}. The final statement follows from Example~\ref{ex:3rev}. 
\end{proof}

\section{nonlinear system characterizations}\label{sec:NL}
Classical notions of \emph{passive} and \emph{negative imaginary} systems are closely related to phase properties. Since the signed scaled graph in \eqref{eq:SGH} provides signed phase information, it is natural to consider this definition in the context of passive and negative imaginary nonlinear systems.  

\subsection{Passive systems}
In classical literature, passive systems are typically thought of having their phase contained in $[-\pi/2, \pi/2]$. Considering the definition of nonlinear system phase in \eqref{eq:phase}, this naturally results in the following definition.

\begin{definition}[Passive systems]\label{def:PR}
    A stable system $H : \mathcal{L}_2 \to \mathcal{L}_2$ is called input strictly passive if for all $u \in \mathcal{L}_2$ and all $y = H(u) \in \mathcal{L}_2$, there exists $\varepsilon >0$ such that
    \begin{equation}\label{eq:PR}
        \langle u,y\rangle \geq \varepsilon \|u\|^2.
    \end{equation}
    When $\varepsilon =0$ we call the system passive.
\end{definition}
Definition~\ref{def:PR} is consistent with the existing definitions of passive systems in the literature, see, for instance, \cite{vdSchaft17, Khalil02}. Clearly, the SSG of a (strictly) passive system belongs to the (open) right-half complex plane. The following result shows that the classical passivity theorem is a special case of Theorem~\ref{th:SGH}.

\begin{corollary}
The negative feedback interconnection of an input strictly passive system $H_1:\mathcal{L}_2 \to \mathcal{L}_2$ and a passive system $H_2:\mathcal{L}_2 \to \mathcal{L}_2$ is finite-gain stable for all $w\in \mathcal{L}_2$. 
\end{corollary}

\begin{proof}
First note that for $z \in \textup{SSG}(H)$ we have $\textup{Re}\left\{z\right\} = \langle u,y\rangle/\|u\|^2$. From Definition~\ref{def:PR} it follows that $$\textup{SSG}(H_1) \subset \left\{z \in \mathbb{C} \mid \textup{Re}\left\{z \right\}\geq \varepsilon>0\right\}$$ and $$-\tau \textup{SSG}^\dagger(H_2) \subset \left\{z \in \mathbb{C} \mid \textup{Re}\left\{z \right\} \leq 0\right\},$$ where the minus sign results form the negative feedback interconnection. Hence, the graph separation condition \eqref{eq:cond2} in Theorem~\ref{th:SGH} is automatically satisfied. Besides, since $H_1$ is strictly passive and $H_2$ are passive, the situation $\langle u_1,y_1\rangle = -\langle u_2,y_2\rangle$ cannot happen for nontrivial $u_1\neq 0$, and the set $\mathcal{W}$ is empty.
\end{proof}

\subsection{Negative imaginary systems}
In a similar spirit as passive systems, negative imaginary systems are usually understood to have their phase contained in $[-\pi,0]$. Based on this understanding, we introduce an alternative definition of nonlinear negative imaginary systems.
\begin{definition}[SSG-Negative imaginary systems]\label{def:NI}
     A stable system $H : \mathcal{L}_2 \to \mathcal{L}_2$ is said to be SSG-negative imaginary if for all $u \in \mathcal{L}_2$ and $y = H(u) \in \mathcal{L}_2$, there exists $\varepsilon >0$ such that
      \begin{equation}\label{eq:niuy}
        \pmb{\langle} \hat{u},\mathcal{E}(y)\pmb{\rangle} \geq \varepsilon(\|u\|\|y\| - |\langle u,y\rangle|).
    \end{equation}      
\end{definition}
Definition~\ref{def:NI} is inspired by our definition of the SSG. That is, condition \eqref{eq:niuy} implies that all points $z \in \textup{SSG}(H)$ belong to $\{z\in\mathbb{C}\;|\;\text{Im}(z)\leq 0\}$. Indeed, from the Cauchy-Schwarz inequality, the right-hand side of \eqref{eq:niuy} is non-negative, and is zero if $\pmb{\langle} \hat{u},\mathcal{E}(y)\pmb{\rangle} = 0$. In that case, the excessive points as discussed in Remark~\ref{rem:1} collapse to a single point on the real axis. Hence, the SSG of an SSG-negative imaginary system belongs to the lower half of the complex plane. We coin the term SSG-negative imaginary system to differentiate from the existing definitions of negative imaginary systems in the literature. That is, in the classical definition of negative imaginary systems the condition $\langle u,\dot{y} \rangle \geq 0$ is used, and the additional term in the right-hand side of \eqref{eq:niuy} is absent; see, e.g., \cite{Ghallab18, Lanzon23} for definitions. In the LTI case, however, Definition~\ref{def:NI} turns out to be consistent with those in the literature. By taking inputs $u \in \mathcal{L}_2$ having their Fourier spectrum $u(j\omega)$ centered on frequency $\pm \omega$ only and $\|u\|=1$, \eqref{eq:niuy} implies that $$\textup{Im}\left\{H(j\omega)\right\} \leq \varepsilon \left(|\textup{Re}\left\{H(j\omega)\right\}|-|H(j\omega)|\right),$$ see also the discussion around \eqref{eq:uyl} and \eqref{eq:uhyl}. From the fact that $|H(j\omega)| = (\textup{Re}\left\{H(j\omega)\right\}^2 + \textup{Im}\left\{H(j\omega)\right\}^2)^{\frac{1}{2}}$ the above inequality implies $\textup{Im}\left\{H(j\omega)\right\} \leq 0$ for all $\omega \in [0,\infty]$. The next result provides an analog of the classical negative imaginary theorem on the basis of scaled graphs. 




\begin{theorem}\label{col:NI}
    The positive feedback interconnection of two SSG-negative imaginary systems $H_1: \mathcal{L}_2 \to \mathcal{L}_2$ and $H_2: \mathcal{L}_2 \to \mathcal{L}_2$ is finite-gain stable for all $w \in \mathcal{L}_2 \setminus \mathcal{W}$ if    \begin{equation}\label{eq:z1z2}
        z_1 z_2 < 1
    \end{equation}
    for all $z_1 \in \textup{SSG}(H_1) \cap \mathcal{X}$ and all  $z_2 \in \textup{SSG}(H_2) \cap \mathcal{X}$,
    where $\mathcal{X}=\left\{z \in \mathbb{C} \mid \textup{Im}\left\{z\right\} = 0\right\}$.
\end{theorem}

\begin{proof}
    First, note that from Definition~\ref{def:NI} it follows that 
    $$\textup{SSG}(H_1) \subset \left\{z \in \mathbb{C} \mid \textup{Im}\left\{z \right\} \leq 0\right\}$$ 
    and
     $$\textup{SSG}^\dagger(H_2) \subset \left\{z \in \mathbb{C} \mid \textup{Im}\left\{z \right\} \geq 0\right\},$$
     which follows from $ \pmb{\langle} \hat{u},\mathcal{E}(y)\pmb{\rangle} \geq 0$, and in case $ \pmb{\langle} \hat{u},\mathcal{E}(y)\pmb{\rangle} = 0$, \eqref{eq:niuy} implies $|\langle u,y\rangle| = \|u\|\|y\|$, such that $z = \|y\|/\|u\|e^{\pm j\theta(u,y)} = \pm \langle u,y\rangle/\|u\|^2$, i.e., the imaginary part is zero in this case. Note that the above inequality for $\textup{SSG}^\dagger (H_2)$ follows from the fact that in the inverse of $\textup{SSG}(H_2)$, the phase flips sign. By condition \eqref{eq:z1z2}, it follows that $z_1 \neq z_2^\dagger$ for all $z_i \in \mathcal{X}$ with $i\in \left\{1,2\right\}$, and all scalings $\tau \in (0,1]$. Together with the separation of $\textup{SSG}(H_1) \setminus \mathcal{X}$ and $\textup{SSG}^\dagger(H_2)\setminus \mathcal{X}$ which also holds for all scalings $\tau \in (0,1]$, it follows that condition~\eqref{eq:cond2} in Theorem~\ref{th:SGH} is satisfied.  
\end{proof}
Finally, note that the interconnection in Example~\ref{ex:3rev} in fact can be seen as the positive feedback interconnection of two SSG-negative imaginary systems.

\section{Conclusions}\label{sec:conclusion}
In this paper, we provided a modified definition of the scaled graph (SG), in which phase is characterized through the Hilbert transform. Contrary to the original definition of the SG, the modified SG contains \emph{signed} phase information, thereby removing the inherent symmetry in the SG. This allows us to differentiate between, e.g., lead and lag filters, which have no discernible difference in the classical definition of the SG, and in that sense potentially introduce conservatism in SG analysis. We have shown how the new SG leads to intuitive characterizations of passive and negative-imaginary nonlinear systems, and provided various interconnection results. The ideas put forward in this paper enrich the theory on scaled graphs and provide the next steps in graphical nonlinear system analysis and design. For future work we aim at deriving interconnection rules, as well as developing methods for efficiently generating the signed scaled graph of nonlinear systems.

\appendix




\subsection{Proof of Theorem~\ref{th:SGH}}\label{app:pf2}
The proof exploits a homotopy argument \cite{Megretski, Freeman2022} in which we consider a collection of feedback systems, being the interconnection of $H_1$ and $\tau H_2$, and show that for all $\tau \in (0,1]$, there exists a constant $c_0>0$ independent of $\tau$ such that $\|y\| \leq c_0\|w\|$, with $y = (y_1, y_2)^\top$.

Recall that $z_1 \in \textup{SSG}(H_1)$ is written as $$z_1 = \frac{\|y_1\|}{\|u_1\|}e^{js_1(u_1,y_1)\theta_1(u_1,y_1)},$$ where $s_1(u_1,y_1) = \textup{sgn}( \pmb{\langle} \hat{u}_1,\mathcal{E}(y_1)\pmb{\rangle})$ and $z_2 \in \textup{SSG}^\dagger(-\tau H_2)$ is written as $$z_2 = \frac{\|-\tau y_2\|}{\|u_2\|}e^{js_2(u_2,-y_2)\theta_2(u_2, -y_2)},$$ where $$s_2(u_2,-y_2) = -\textup{sign}( -  \pmb{\langle} \hat{u}_2,\mathcal{E}(y_2)\pmb{\rangle} ) = \textup{sign}( \pmb{\langle} \hat{u}_2,\mathcal{E}(y_2)\pmb{\rangle} ).$$ Note that $\tau$ and the minus sign appear since we consider $z_2 \in  \textup{SSG}^\dagger(-\tau H_2)$, i.e., the output of $H_2$ is scaled with a gain $-\tau$.

Under the hypothesis of the theorem, for all $z_1 \in \textup{SSG}(H_1)$ and all $z_2 \in \textup{SSG}^\dagger(-\tau H_2)$, we have $|z_1 -z_2| \geq r > 0$. Hence, this is equivalent to either
\begin{enumerate}
    \item $\arg (z_1) = \arg(z_2)$ implies $||z_1|-|z_2||\geq \epsilon$ or 
    \item $|z_1|=|z_2|$ implies $|\arg(z_1) -\arg(z_2)|\geq \epsilon$,
\end{enumerate}
where $\epsilon >0$ is a uniform constant that depends on $r$. We proceed by considering both cases separately.

Case $1)$: $\arg (z_1) = \arg(z_2)$ implies $||z_1|- |z_2||\geq \epsilon$. First, suppose $|z_1| \leq |z_2|-\epsilon$, which, by definition implies
\begin{equation}\label{eq:small_gain}
   \tau  \frac{\|y_1\|}{\|u_1\|}\frac{\| y_2\|}{\|u_2\|} \leq  1-\epsilon \frac{\| y_2\|}{\|u_2\|} < 1
\end{equation}
for all $u_1, u_2$ such that $|z_1| \leq |z_2|-\epsilon$.
That is, without loss of generality, take 
$$
\frac{\|y_1\|}{\|u_1\|} \leq \gamma_1, \textup{ and } \tau \frac{\|y_2\|}{\|u_2\|} \leq \gamma_2,
$$
with $\gamma_1, \gamma_2 > 0$ and $\gamma_1 \gamma_2 <1$. Since $u_1 = w-\tau y_2$ and $u_2 = y_1$ we find

\begin{align*}
    \|u_1\| \leq \gamma_2 \|u_2\| + \|w\|, \textup{ and }    \tau\|u_2\| \leq \gamma_1 \|u_1\|.
\end{align*}
Since $\gamma_1 \gamma_2 <1$ we eventually find
\begin{align*}
    \|u_1\| \leq \frac{1}{1-\gamma_1 \gamma_2}\|w\|, \quad \textup{and}\quad \|u_2\|  \leq \frac{\gamma_1}{1- \gamma_1 \gamma_2}\|w\|.
\end{align*}
Using boundedness of the operators $H_1$ and $H_2$, there exist $r_1,r_2 < \infty$ such that $\|y_1\|\leq r_1\|u_1\|$ and $\|y_2\|\leq r_2\|u_2\|$. Hence, we find 
\begin{align*}
    \|y_1\| &\leq \frac{r_1}{1-\gamma_1 \gamma_2}\|w\|, \quad \textup{and} \quad \|y_2\| & \leq \frac{\gamma_1 r_2}{1- \gamma_1 \gamma_2}\|w\|.
\end{align*}
As such, we find for all $\tau \in (0,1]$
\begin{equation}
    \|y\| \leq \frac{r_1+\gamma_1r_2}{1-\gamma_1 \gamma_2} \|w\| := c_1 \|w\|.
\end{equation}
Similarly, consider $|z_2|\leq|z_1|-\epsilon$. In this case, we find \begin{equation}\label{eq:small_gain2}
    \tau \frac{\|y_1\|}{\|u_1\|}\frac{\|y_2\|}{\|u_2\|} \geq 1 + \epsilon  \frac{\|y_2\|}{\|u_2\|} > 1
\end{equation}
for all $u_1, u_2$ such that $|z_2| \leq |z_1|-\epsilon$ and we can apply the large-gain theorem. That is, w.l.o.g. we can take 
$$
\frac{\|y_1\|}{\|u_1\|} \geq \mu_1, \textup{ and } \tau \frac{\|y_2\|}{\|u_2\|} \geq \mu_2,
$$
with $\mu_1, \mu_2 > 0$ and $\mu_1 \mu_2 > 1$. Using $u_1 = w-\tau y_2$ and $u_2 = y_1$ we find
\begin{align*}
    \mu_1 \|u_1\| &\leq \|y_1\| = \|u_2\|, \\
    \mu_2\|u_2\| &\leq \tau \|y_2\| \leq \|w\|+\|u_1\|.
\end{align*}
Using $\mu_1\mu_2>1$, from the above inequalities we find
\begin{align*}
    \|u_1\| & \leq \frac{1}{\mu_1 \mu_2 -1}\|w\|, \\
    \|u_2\| & \leq \frac{\mu_1}{\mu_1 \mu_2 -1}\|w\|.
\end{align*}
Again, using boundedness of $H_1$ and $H_2$ we find for all $\tau \in (0,1]$ 
\begin{equation}
    \|y\| \leq \frac{r_1+\mu_1 r_2}{\mu_1\mu_2-1}\|w\| := c_2 \|w\|.
\end{equation}

Case $2)$: $|z_1|=|z_2|$ implies $|\arg(z_1) -\arg(z_2)|\geq \epsilon$. We can distinguish two sub cases: $2a)$ $|\arg(z_1)| \neq |\arg(z_2)|$, and $2b)$ $\arg(z_1) = -\arg(z_2)$. We treat these cases separately.

Case $2a)$: $|\arg(z_1)| \neq |\arg(z_2)|$. In this case, the sign does not matter and we can simply consider the unsigned angle $\theta(u,y) = \arccos ({\langle u,y\rangle}/{\|u\|\|y\|}) \in [0,\pi]$. First, consider the case $\theta_1 < \theta_2$. Equivalently, there must exist
$\delta >0$ such that $\cos(\theta_1)-\cos(\theta_2) \geq \delta >0$ since $\cos(\cdot)$ is monotonically decreasing on $[0,\pi]$. By definition of the singular angle, we find for all $\tau \in (0,1]$
\begin{equation}\label{eq:i1}
    \frac{\langle u_1, y_1\rangle}{\|u_1\|\|y_1\|} +  \frac{\langle u_2, \tau y_2\rangle}{\|u_2\|\|\tau y_2\|} \geq \delta >0
\end{equation}
for all $u_1, u_2$ such that $|\arg(z_1)| < |\arg(z_2)|$. Since $u_1 = w-\tau y_2$ and $u_2 = y_1$ we find
\begin{equation}\label{eq:i2}
   \frac{\langle w, y_1\rangle}{\|u_1\|\|y_1\|} - \frac{\langle \tau y_2, y_1\rangle}{\|u_1\|\|u_2\|} +  \frac{\langle y_1, \tau y_2\rangle}{\|y_1\|\|\tau y_2\|} \geq \delta >0.
\end{equation}
By assumption (since $|z_1|=|z_2|$) we have $\|y_1\|/\|u_1\| = \|u_2\|/\|\tau y_2\|$, which implies that 
$\|y_1\|\|\tau y_2\| = \|u_1\|\|u_2\|$. As such, \eqref{eq:i2} reduces to
\begin{equation}\label{eq:i3}
     {\langle w, y_1\rangle} \geq \delta {\|u_1\|\|y_1\|}.
\end{equation}
Applying Cauchy-Schwarz to the left-hand side of \eqref{eq:i3} gives
\begin{equation}
    \|w\|\geq \delta \|u_1\|.
\end{equation}
By boundedness of $H_1$, $\|y_1\|/\|u_1\| \leq r_1< \infty$, we find
\begin{align*}
    \|y_1\| \leq \frac{r_1}{\delta} \|w\|, \textup{ and } \|y_2\| \leq \frac{r_1r_2}{\delta} \|w\|,
\end{align*}
and, hence, for all $\tau \in (0,1]$
\begin{equation}
    \|y\| \leq \frac{r_1(1+r_2)}{\delta } \|w\| := c_3 \|w\|.
\end{equation}
In a similar way, consider the case $\theta_2 < \theta_1$, which implies the existence of $\delta >0$ such that
for all $\tau \in (0,1]$
\begin{equation}\label{eq:i10}
    -\frac{\langle u_2, \tau y_2\rangle}{\|u_2\|\|\tau y_2\|} -  \frac{\langle u_1, y_1\rangle}{\|u_1\|\| y_1\|} \geq \delta >0.
\end{equation}
Since $u_1 = w-\tau y_2$ and $u_2 = y_1$ we find in this case
\begin{equation}\label{eq:i21}
   -\frac{\langle w, y_1\rangle}{\|u_1\|\|y_1\|} - \frac{\langle \tau y_2, y_1\rangle}{\|y_1\|\|\tau y_2\|} +  \frac{\langle y_1, \tau y_2\rangle}{\|u_1\|\| u_2\|} \geq \delta >0.
\end{equation}
By assumption (since $|z_1|=|z_2|$) we have $\|y_1\|/\|u_1\| = \|u_2\|/\|\tau y_2\|$, or equivalently 
$\|y_1\|\|\tau y_2\| = \|u_1\|\|u_2\|$, leading to $\|w\| \geq \delta \|u_1\|$. Similar as before, this eventually gives
\begin{equation}
    \|y\| \leq \frac{r_1(1+r_2)}{\delta} \|w\| := c_4 \|w\|.
\end{equation}
Case $2b)$: $\arg(z_1) = -\arg(z_2)$. More precisely, $z_1=z_2^*$. In this case we have that $\theta_1(u_1, y_1) = \theta_2(-y_2, u_2)$, but because of the signed quantity we find that $\pmb{\langle} \hat{u}_1,\mathcal{E}(y_1)\pmb{\rangle}$ and $\pmb{\langle} \tau \hat{y}_2,\mathcal{E}(u_2)\pmb{\rangle}$ have opposite signs. Again, this yields a few sub cases. The case where both quantities are zero reduces to the unsigned scaled graph case, to which the analysis in case $2a$ applies; we therefore consider the cases where at least one of the quantities is non-zero. Without loss of generality, assume that $\pmb{\langle} \hat{u}_1,\mathcal{E}(y_1)\pmb{\rangle} > 0$ and $\pmb{\langle} \tau \hat{y}_2,\mathcal{E}(u_2)\pmb{\rangle}<0$. For all $w \in \mathcal{L}_2\setminus \mathcal{W}$ we can rule out this case from happening in the feedback interconnection.

Combining all above cases leads to the conclusion that there exists $c_0=\max_{i=1}^4 c_i>0$ independent of $\tau$ such that for all $w \in \mathcal{L}_2$ and all $\tau \in (0,1]$ we have $\|y\| \leq c_0 \|w\|$. Under the assumption of well-posedness, the result then follows from \cite[Thm. 3.2]{Freeman2022}.

\end{document}